\documentclass[12pt]{article}

\usepackage{amssymb, amsmath, amsthm, tikz, verbatim, graphicx,lmodern,indentfirst,enumerate,color}

\usepackage[titletoc,page]{appendix}
\usepackage[margin=1.4in]{geometry}

\usepackage{cite}
\usepackage[skip=0pt]{caption}

\newtheorem{thm}{Theorem}[section]
\newtheorem{lem}[thm]{Lemma}
\newtheorem{prop}[thm]{Proposition}

\theoremstyle{definition}

\theoremstyle{definition}
\newtheorem{defn}[thm]{Definition}

\theoremstyle{remark}
\newtheorem{rem}[thm]{Remark}

\numberwithin{equation}{section}

\makeatletter

\newcommand{\Rmnum}[1]{\expandafter\@slowromancap\romannumeral #1@}
\newcommand\restr[2]{{
  \left.\kern-\nulldelimiterspace 
  #1 
  \right|_{#2} 
  }}

\makeatother
\begin{document}

\title{Odd period cycles and ergodic properties in price dynamics for an exchange economy}
\author{Tomohiro Uchiyama\\
Faculty of International Liberal Arts, Soka University,\\ 
1-236 Tangi-machi, Hachioji-shi, Tokyo 192-8577, Japan\\
\texttt{Email:t.uchiyama2170@gmail.com}}
\date{}
\maketitle 

\begin{abstract}
In the first part of this paper (Sections~1-4), we study a standard exchange economy model with Cobb-Douglas type consumers and give a necessary and sufficient condition for the existence of an odd period cycle in the Walras-Samuelson (tatonnement) price adjustment process.  We also give a sufficient condition for a price to be eventually attracted to a chaotic region. In the second part (Sections~5 and~6), we investigate ergodic properties of the price dynamics showing that the existence of chaos is not necessarily bad. (The future is still predictable on average.) Moreover, supported by a celebrated work of Avila et al. (Invent.~Math., 2003), we conduct a sensitivity analysis to investigate a relationship between the ergodic sum (of prices) and the speed of price adjustment. We believe that our methods in this paper can be used to analyse many other chaotic economic models.  
\end{abstract}

\section{Introduction}
In this paper, we study a standard exchange economy model with two consumers of Cobb-Douglas type and two goods. Denote the two consumers by $i =1,2$ and the two goods by $x$ and $y$. Let $\alpha, \beta \in (0,1)$. We assume that the consumer $1$ has the utility function $u^1(x,y)=x^\alpha y^{1-\alpha}$ and the consumer $2$ has the utility function $u^2(x,y)=x^{\beta} y^{1-\beta}$ in the consumption space $\mathbb{R}^2_{+}$. We also assume that the consumer $1$ has the initial endowment $w^1=(\bar x,0)$ and the consumer $2$ has the initial endowment $w^2=(0,\bar y)$ where $\bar x, \bar y >0$. Fix the price of good $y$ as $1$ and that of good $x$ as $p>0$. Then by the standard optimisation result under the budget constraints, we obtain the excess demand function $z(p)$ for good $x$, that is given by $z(p)=\bar y\beta/p - \bar x(1-\alpha)$. Now we define the Walras-Samuelson (tatonnement) price adjustment process by
\begin{equation}
p_{t+1} = f(p_t) = p_{t}+\lambda z(p_t) = p_{t}+\lambda[\bar y\beta/p_t-\bar x(1-\alpha)] \textup{ where } \lambda\in \mathbb{R}_{++}. \label{dynamics}
\end{equation} 
Note that $\lambda$ denotes the speed of adjustment and $p_j$ denotes the price of good $x$ at time $j$.\\

Let $E=(0,1)\times (0,1)$. Then the following (that is a slight extension of~\cite[Prop.~9.10]{Bhattacharya-book}) is not difficult to show (see Section~4 for a proof):
\begin{prop}\label{LiYorke}
There exist open sets $A\subset E$, $B\subset \mathbb{R}^2_{++}$, and $C\subset \mathbb{R}_{++}$ such that if $(\alpha, \beta)\in A$, $(\bar x, \bar y)\in B$, and $\lambda\in C$ then the process (\ref{dynamics}) has a period three cycle (hence exhibits a Li-Yorke chaos). 
\end{prop}
It is well-known that the existence of a period three cycle implies that of a Li-Yorke chaos (by the famous Li-Yorke theorem~\cite[Thm.~1]{LiYork-PeriodThree-Month}), and this argument has been used a lot in economic literature, see~\cite{BenhabibDay-Erratic-Letters},~\cite{Benhabib-Erratic-JEDC},~\cite{DayShafer-Keynesian-Macro},~\cite{NishimuraYano-Prog-soliton} for example. However, this is a bit overkill: by~\cite[Chap.~\Rmnum{2}]{Block-book}, we know that the existence of a cycle of \emph{any} odd length (not necessarily of period three) implies that of a Li-Yorke chaos. In the first part of this paper, extending Proposition~\ref{LiYorke}, we obtain:

\begin{thm}\label{main}
Let $\frac{\bar y\beta}{{\bar x}^2(1-\alpha)^2} < \lambda < \frac{4\bar y\beta}{{\bar x}^2(1-\alpha)^2}$. Then the map $f$ in Equation (\ref{dynamics}) has the following properties:
\begin{enumerate}
\item{$\restr{f}{E}\in \mathfrak{G}$.}
\item{$\restr{f}{E}$ has an odd period cycle if and only if $\frac{25 \bar y\beta}{9 {\bar x}^2 (1-\alpha)^2} < \lambda < \frac{4\bar y\beta}{{\bar x}^2 (1-\alpha)^2}$.}
\item{The second iterate $(\restr{f}{E})^2$ is turbulent if and only if $\frac{25 \bar y\beta}{9 {\bar x}^2 (1-\alpha)^2} \leq \lambda < \frac{4\bar y\beta}{{\bar x}^2 (1-\alpha)^2}$.} 
\item{If $\frac{25 \bar y\beta}{9 {\bar x}^2 (1-\alpha)^2} \leq \lambda < \frac{4\bar y\beta}{{\bar x}^2 (1-\alpha)^2}$, then the closed interval $E$ is attracting for $f$, that is, $f^n(p)\in E$ for some $n\in \mathbb{N}$.} \end{enumerate}
\end{thm}

Our first main result of this paper (Theorem~\ref{main}) extends Proposition~\ref{LiYorke} in the following way: 1.~Replace a period three cycle with an odd period cycle (or a turbulence for the second iterate), 2.~Give a specific (algebraic) form of the sets $A$, $B$, and $C$, 3.~Give a \emph{necessary and sufficient} condition for the existence of an odd period cycle (or a turbulence for the second iterate) rather than giving a sufficient condition only (that is usually done in many economic literature such as~\cite{BenhabibDay-Erratic-Letters},~\cite{Benhabib-Erratic-JEDC},~\cite{DayShafer-Keynesian-Macro},~\cite{NishimuraYano-Prog-soliton}). 

We defer detailed explanations, (mathematical/economic) interpretations, and all the necessary definitions (such as $E$, $\mathfrak{G}$, and "turbulent") for Theorem~\ref{main} to Section~2. Here, we just note two things: (1)~$E$ is some compact interval in $\mathbb{R}$ and $\mathfrak{G}$ is a "nice" class of continuous unimodal maps. (2)~The upper bound for $\lambda$ directly follows from the requirement $p>0$. So, the upper bound in Theorem~\ref{main} is not really interesting. The real meat is in the lower bounds.

In the second part of this paper (Section 5 onwards), we investigate the price dynamics defined by Equation (\ref{dynamics}) using a probabilistic method. The upshot of Theorem~\ref{main} is that if $\lambda$ (the speed of price adjustment) is large enough (with some bound to keep $p>0$), the price dynamics show chaotic behaviours for (uncountable) many initial $p$ (see the definition of a Li-Yorke chaos in Section 2). This means that it is hard to predict the future ($f^n(p)$ for large $n$) for these $p$. This sounds pretty bad, but the main results in the second part of the paper  (Theorem~\ref{finalThm}) show that this not necessarily so. Roughly speaking, we can still "predict" the future on average. 

Our (numerical/theoretical) argument in the second part of the paper use ergodic theory (a quick overview of ergodic theory is given in Section 5 to make the paper self-contained). What we need now is a bare minimum: we just need one definition.
\begin{defn}
$\lim_{n\rightarrow \infty}\frac{1}{n}\sum_{k=0}^{n-1}f^{k}(p)$ is called the \emph{ergodic sum of $f$ with respect to $p$}. (If the limit exists, the ergodic sum of $f$ is same for almost all $p$, so we omit "with respect to $p$".)
\end{defn}

We are ready to state our main result in the second part of the paper. Here, to simplify the exposition and to obtain sharp numerical results, we fix $(\alpha,\beta)=(0.75,0.5)$, $(\bar x, \bar y)=(4,2)$. A similar analysis (as follows) can be done for any $\alpha, \beta, \bar x, \bar y$ (our choice of the parameters is completely ad hoc and our method is quite generic). We obtain

\begin{thm}\label{finalThm}
For $1<\lambda<4$, the ergodic sums of $f$ are as in Figure~\ref{fig15} (possibly except some $\lambda$ values whose total Lebesgue measure is $0$). 
\end{thm}

\begin{figure}[h!]
	\begin{center}
    	\includegraphics[scale=0.6]{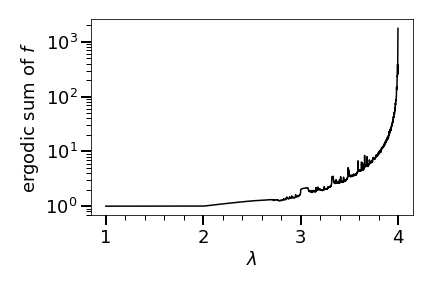}
	\end{center}
    \caption{Ergodic sums of $f$}\label{fig15}
\end{figure}

We defer the detailed explanations/comments on Theorem~\ref{finalThm} to Sections 5 and 6, but to be honest, we found it very surprising that the ergodic sums of $f$ change so smooth and stable (except a few bumps) as $\lambda$ increases considering the fact (as can be seen from the bifurcation diagram of $f$,~Figure~\ref{fig12}) that as $\lambda$ increases $f^n(p)$ go through a quite a bit of qualitative changes from a stable fixed point, attracting periodic orbit of different periods, and finally chaotic behaviours.

Here is the structure of the rest of the paper. In Section~2, we give necessary definitions involving chaos and explain how we got our first main theorem (Theorem~\ref{main}) together with its (mathematical/economic) interpretations. Next, in Section~3, we give the proofs of Theorem~\ref{main} and Proposition~\ref{LiYorke}. In Section~4, after giving a quick overview of ergodic theory, we explain a generic method to obtain our second main result (Theorem~\ref{finalThm}) with its important stepping stone (Theorem~\ref{ergThm1}). Our argument here is a delicate combinations of results from ergodic theory and numerical computations. In Section~5, we give some (mathematical) interpretations of Theorem~\ref{finalThm} including its connection to a deep result in ergodic theory (Proposition~\ref{Avila}) of Avila et.al.  

All programming files used for numerical calculations and for generating plots in this paper (Jupyter notebooks) are available upon request.

\section{Odd period cycles}
\subsection{Definitions involving a chaos}
First, we clarify what we mean by a Li-Yorke chaos, a turbulence, and a topological chaos. (There are several definitions of a chaos in the literature.) The following definitions are taken from~\cite[Def.~5.1]{Ruette-book} and~\cite[Chap.~\Rmnum{2}]{Block-book}. Let $g$ be a continuous map of a closed interval $I$ into itself: 

\begin{defn}
We say that $g$ exhibits a \emph{Li-Yorke chaos} if there exists an uncountable \emph{scrambled set} $S\subset I$, that is, for any $x,y\in S$ we have 
\begin{equation*}
\limsup_{n\rightarrow \infty}\mid g^n(x)-g^n(y)\mid > 0 \textup{ and } \liminf_{n\rightarrow \infty}\mid g^n(x)-g^n(y)\mid=0,
\end{equation*}
and for $x\in S$ and $y$ being a periodic point of $g$,
\begin{equation*}
\limsup_{n\rightarrow \infty}\mid g^n(x)-g^n(y)\mid > 0.
\end{equation*}
\end{defn}

\begin{defn}
We call $g$ \emph{turbulent} if there exist three points, $x_1$, $x_2$, and $x_3$ in $I$ such that $g(x_2)=g(x_1)=x_1$ and $g(x_3)=x_2$ with either $x_1<x_3<x_2$ or $x_2<x_3<x_1$. Moreover, we call $g$ \emph{(topologically) chaotic} if some iterate of $g$ is turbulent.  
\end{defn}
It is known that a map $g$ is topologically chaotic if and only if $g$ has a periodic point whose period is not a power of $2$, see~\cite[Chap.~\Rmnum{2}]{Block-book}. This implies that a map $g$ is topologically chaotic if and only if the topological entropy of $g$ is positive, see~\cite[Chap.~\Rmnum{8}]{Block-book}. In the first part of this paper, we focus on a topological chaos (or a positive topological entropy) in the context of price dynamics. See~\cite{Ruette-book} for more characterisations and (subtle) mutual relations of various kinds of chaos.

\subsection{Explanations and Interpretations of Theorem~\ref{main}} 

Here, we explain a key result to obtain Theorem~\ref{main}. Along the way, we  give definitions for all the unexplained notation (such as $\mathfrak{G}$ and $E$) in Theorem~\ref{main}. First, we recall the following mathematical result characterising the existence of a topological chaos for a unimodal interval map~\cite[Cor.~3]{Deng-TopChaos-JET}. Theorem~\ref{main} is a (highly non-trivial as seen in Section 3) consequence (or a special case) of~\cite[Cor.~3]{Deng-TopChaos-JET}. Let $\mathfrak{G}$ be the set of continuous maps from a closed interval $[a, b]$ to itself so that an arbitrary element $g\in \mathfrak{G}$ satisfies the following two properties:
\begin{enumerate}
\item{there exists $m\in (a,b)$ with the map $g$ strictly decreasing on $[a,m]$ and strictly increasing on $[m,b]$.}
\item{$g(a)>a$, $g(b)\leq b$, and $g(x)<x$ for all $x\in[m,b)$.}
\end{enumerate}
For $g\in \mathfrak{G}$, let $\Pi:=\{x\in [a,m]\mid g(x)\in [a,m] \textup{ and } g^2(x)=x\}$. Now we are ready to state~\cite[Cor.~3]{Deng-TopChaos-JET}:
\begin{prop}\label{ChaosThm}
Let $g\in \mathfrak{G}$. The map $g$ has an odd-period cycle if and only if $g^2(m) > m$ and $g^3(m) > \max\{x\in \Pi\}$ and the second iterate $g^2$ is turbulent if and only if $g^2(m) > m$ and $g^3(m) \geq \min\{x\in \Pi\}$. 
\end{prop}

We keep the same notation $f$, $\bar x$, $\bar y$, $\alpha$, $\beta$, and $\lambda$ from Equation (\ref{dynamics}). We write $\restr{f}{E}$ for the restriction of $f$ to the closed interval $E:=[f(\sqrt{\bar y\lambda\beta}),f^2(\sqrt{\bar y\lambda\beta})+\sqrt{\bar y\lambda\beta}]$. (We will explain the significance of the number $\sqrt{\bar y\lambda\beta}$ and the reason for the choice of $E$ in the next section.) Note that in the next section we will show that $f(\sqrt{\bar y\lambda\beta})<f^2(\sqrt{\bar y\lambda\beta})+\sqrt{\bar y\lambda\beta})$ (so, $E$ is a non-degenerate closed interval) and that $\restr{f}{E}$ is a map to $E$.

Now, we give several comments/interpretations on Theorem~\ref{main}. First, in the next section, we will show that the condition $\frac{\bar y\beta}{{\bar x}^2(1-\alpha)^2} < \lambda < \frac{4\bar y\beta}{{\bar x}^2(1-\alpha)^2}$ in Theorem~\ref{main} is the weakest condition for $f$ to be economically meaningful (to keep $p>0$) and also for $\restr{f}{E}$ to be in $\mathfrak{G}$. As stated in Introduction, the upper bound for $\lambda$ is not really interesting and the real meat is in the lower bounds in parts 2,3, and 4. 

Second, it is well-know that the vertical stretch the graph of $f$ controls the existence of a chaos: if we stretch the graph further, we are likely to obtain a chaos. Now, looking at Equation (\ref{dynamics}), we see that if we make $\alpha$, $\beta$, $\bar y$ small, or $\bar x$ large, the "valley" of the graph of $f$ goes deep down. Also, it is clear that a small $\lambda$ makes the graph of $f$ "flat". Our lower bound $L:=\frac{25 \bar y\beta}{9 {\bar x}^2 (1-\alpha)^2}$ agrees with these observations: $L$ is an increasing function of $\alpha$, $\beta$, $\bar y$ and a decreasing function of $\bar x$. In other words, it is easy to generate a chaos (small $\lambda$ gives a chaos) if $\alpha$, $\beta$, or $\bar y$ is small or $\bar x$ is large. We give an economic explanation for this. By looking at the form of the excess demand function $z(p)=\bar y\beta/p - \bar x(1-\alpha)$ (or the demand function of each consumer for good $x$, that is, $x=\alpha \bar x$ for consumer 1 and $x=\frac{\beta \bar y}{p}$ for consumer 2 respectively), we see that $p$ goes up very sharply (almost irrespective of the parameter values) after $p$ gets close to $0$. To generate a chaos, $p$ needs to drop sharply afterwords, that is possible (or at least easy) if $\alpha$, $\beta$, or $\bar y$ is small or $\bar x$ is large for the following reasons: 1.~if $\alpha$ or $\beta$ is small, the demand for good $1$ is weak (thus $p$ drops sharply), 2.~if $\bar y$ is small, then the demand of consumer 2 for good $x$ (that is excessively strong when $p$ is close to $0$) drops sharply (since the budget for consumer 2 is tight), 3.~if $\bar x$ is large, when $p$ is very high, a large excess supply happens. 

Third, part 4 of Theorem~\ref{main} shows that if $\lambda$ is sufficiently large, for any initial $p>0$, our price dynamics eventually trap $p$ inside a (Li-Yorke) chaotic region $E$. This is interesting since parts 1, 2, and 3 say nothing about what is going on outside of $E$. This sort of analysis is not done in~\cite{Deng-TopChaos-JET}. 

\subsection{Odd period (but no period three) cycles}

To end this section, we give an application of Theorem~\ref{main}. By the Sharkovsky order in~\cite{Sharkovsky-order-Ukr}, we know that if the map $f$ has a cycle of period three, then it also has a cycle of any odd order. Thus, it is natural to guess that if $\lambda$ is close to the lower bound for $\lambda$ in Theorem~\ref{main} (but still above the lower bound), the map $f$ has an odd period cycle but no period three cycle. We give one example where this is actually the case. Using our concrete characterisation of the existence of an odd period cycle, we obtain:
\begin{prop}\label{NoPeriodThree}
Let $\bar x = 4, \bar y =2, \alpha=0.75, \beta=0.5$. If $2.77<\lambda\leq 3.00$, then the map $f$ in Equation (\ref{dynamics}) has an odd period cycle but no period three cycle. 
\end{prop}
Note that in this case, by Theorem~\ref{main} a necessary and sufficient condition for the existence of an odd period cycle is $\frac{25 \bar y\beta}{9 {\bar x}^2 (1-\alpha)^2}=2.77<\lambda<4=\frac{4\bar y\beta}{{\bar x}^2 (1-\alpha)^2}$. (We use this condition in Section 6.) Our numerical computation (see Section 4 for details) shows that Proposition~\ref{NoPeriodThree} is almost an if and only if statement: for $\lambda\geq 3.01$, we get a period three cycle.  Although many examples of this sort can be obtained by the same method, a complete characterisation for the existence of a period three cycle for a unimodal map is not known. (Thus it is not possible to obtain the precise $\lambda$ value where a bifurcation happens.) We leave it for a future work.

\section{Proof of Theorem~\ref{main}}
In the following proof, most results follow from direct (but a fairly complicated) algebraic calculations. We give some sketches of our manipulations while pointing some important steps out rather than writing all the details. All calculations can be checked by a computer algebra system, say, Magma~\cite{magma}, Python~\cite{10.5555/1593511}, etc.\\

First, for the function $f$ in Equation (\ref{dynamics}), we have $f'(p)=1-\frac{\bar y\lambda\beta}{p^2}$ and $f''(p)=\frac{2\bar y\lambda\beta}{p^3}>0$ for any $p>0$. So, $f$ is strictly convex (unimodal) and takes its minimum at $p=\sqrt{\bar y\lambda\beta}$. We sometimes write $s$ for $\sqrt{\bar y\lambda\beta}$ to ease the notation. First of all, since we assume that $p>0$, we must have $f(p)>0$ for any $p>0$, hence $f(s)>0$. This gives that $\lambda < \frac{4\bar y\beta}{{\bar x}^2(1-\alpha)^2}$.

Next, we want (a restriction of) $f$ to be in $\mathfrak{G}$. Note that any function $f$ in $\mathfrak{G}$ must have some $m$ in its interior of the domain, say $[a,b]$, with $f$ strictly decreasing on $[a,m]$ and strictly increasing on $[m,b]$. Since our $f$ is unimodal, this forces $m$ to be $s$. Moreover $f$ needs to satisfy $f(p)<p$ for all $p\in [s,b)$. So in particular, we must have $f(s)<s$. This implies $\frac{\bar y\beta}{{\bar x}^2(1-\alpha)^2}<\lambda$. Note that $s<f^2(s)+s$ since $f^2(s)>0$ (under the condition $\lambda < \frac{4\bar y\beta}{{\bar x}^2(1-\alpha)^2}$), so we have $f(s)<s<f^2(s)+s$. Now we set $a:=f(s)$ and $b:=f^2(s)+s$. (Now it is clear that $[a,b]$ is non-degenerate.)

\begin{lem}\label{FirstLem}
If $\frac{\bar y\beta}{{\bar x}^2(1-\alpha)^2} < \lambda < \frac{4\bar y\beta}{{\bar x}^2(1-\alpha)^2}$, then 
$\restr{f}{E} \in \mathfrak{G}$.
\end{lem}
\begin{proof}
We need to show three things: (1) $f(a)>a$ and $f(b)\leq b$. (2) $f(p)<p$ for all $p\in [s,b)$. (3) $\restr{f}{E}$ is a map from $E$ to itself. We begin with (2). Let $p\in [s,b)$. Then we have $p-f(p)=-\lambda[\frac{\bar y\beta}{p}-\bar x(1-\alpha)]$. Since $\lambda>0$ and $\frac{\bar y\beta}{p}-\bar x(1-\alpha)\leq \frac{\bar y\beta}{s}-\bar x(1-\alpha)$, it is enough to show that $\frac{\bar y\beta}{s}-\bar x(1-\alpha)<0$. Now $\frac{\bar y\beta}{s}-\bar x(1-\alpha)=\frac{\bar y\beta}{\sqrt{\bar y\lambda\beta}}-\bar x(1-\alpha)$, so $\frac{\bar y\beta}{s}-\bar x(1-\alpha)<0$
is equivalent to $\lambda>\frac{\bar y\beta}{{\bar x}^2(1-\alpha)^2}$, that is certainly true (since this is our assumption). Now we show that the same argument gives $f(b)< b$ in (1). (This strict inequality is stronger than we need here, but we need it to prove part 4 of Theorem~\ref{main}, see the proof of Lemma~\ref{SixthLem} below.) We have $b-f(b) = -\lambda[\frac{\bar y\beta}{f^2(s)+s}-\bar x(1-\alpha)]$. Since $\lambda>0$ and $\frac{\bar y\beta}{f^2(s)+s}-\bar x(1-\alpha)< \frac{\bar y\beta}{s}-\bar x(1-\alpha)$, it is enough to show that $\frac{\bar y\beta}{s}-\bar x(1-\alpha)< 0$. We have already shown that this is true. Next we show that $f(a)>a$. By a direct calculation, we have that $f(a)-a=f(f(s))-f(s)=\lambda[\frac{\bar y\beta}{f(s)}-\bar x(1-\alpha)]>0$ is equivalent to $\left(\sqrt{\lambda}-\frac{\sqrt{\bar y\beta}}{\bar x(1-\alpha)}\right)^2>0$. Now the assumption $\frac{\bar y\beta}{{\bar x}^2(1-\alpha)^2} < \lambda$ forces $f(a)-a>0$. 

To prove (3), we need to show that the maximum value of $\restr{f}{E}$ does not exceed $b=f^2(s)+s$ (that the minimum value of $\restr{f}{E}$, that is $f(s)$, is in $E$ is clear). Since $f$ is unimodal, the maximum of $f$ on $E$ is taken either at $a$ or at $b$. For the first case, we have $f(a)=f(f(s))=f^2(s)<f^2(s)+m=b$. For the second case, we need $f(b)\leq b$, but this is true by part (1) above.   

\end{proof}

We have proved part 1 of Theorem~\ref{main}. We assume $\frac{\bar y\beta}{{\bar x}^2(1-\alpha)^2} < \lambda < \frac{4\bar y\beta}{{\bar x}^2(1-\alpha)^2}$ for the rest of this section (actually for the rest of the paper). Now we are ready to prove parts 2 and 3 of the theorem. In view of Proposition~\ref{ChaosThm} and Lemma~\ref{FirstLem}, we only need to translate two conditions $f^2(m)>m$ and $f^3(m)>\max\{x\in\Pi\}$ (or $f^3(m)\geq \min\{x\in \Pi\}$) in terms of $\bar x$, $\bar y$, $\alpha$, $\beta$, and $\lambda$. First we show that

\begin{lem}\label{ThirdLem}
$f^2(s)>s$ if and only if $\frac{9\bar y\beta}{4{\bar x}^2(1-\alpha)^2}<\lambda$.
\end{lem}
\begin{proof}
Under the condition $\frac{\bar y\beta}{{\bar x}^2(1-\alpha)^2} < \lambda$, a direct computation shows that $f^2(s)>s$ is equivalent to $2{\bar x}^2(1-\alpha)^2\lambda-5\bar x(1-\alpha)\sqrt{\bar y\beta\lambda}+3\beta\bar y>0$. Now the statement follows.  
\end{proof}

Next we show that

\begin{lem}\label{SecondLem}
If $\frac{9\bar y\beta}{4{\bar x}^2(1-\alpha)^2}<\lambda$, then the set $\Pi=\{x\in [a,s]\mid f(x)\in [a,s] \textup{ and } f^2(x)=x\}$ is a singleton, namely $\Pi=\{\frac{\bar y\beta}{\bar x(1-\alpha)}\}$ (the unique fixed point for $\restr{f}{E}$). 
\end{lem} 
\begin{proof}
First, we compute the fixed points of $\restr{f}{E}$. Solving $f(p)=p$ for $p>0$, we obtain $p=\frac{\bar y\beta}{\bar x(1-\alpha)}$. So, $f$ has the unique fixed point, which we name $z$. It is clear that $z<s$ (this follows from $\frac{\bar y\beta}{{\bar x}^2(1-\alpha)^2} < \lambda$) and $a\leq z$ (since $a=f(s)$ is the minimum of $f$). So, we have $z\in \Pi$. Next, we compute the period $2$ points for $f$ on $[a,s]$. Solving $f^2(p)=p$ for $p$ (with Python), we get $p=\frac{\bar y\beta}{\bar x(1-\alpha)}, -\frac{\bar x \alpha \lambda}{2}+\frac{\bar x\lambda}{2}-\frac{1}{2}\sqrt{{\bar x}^2\alpha^2\lambda^2-2{\bar x}^2\alpha\lambda^2-2\bar y\beta\lambda+{\bar x}^2\lambda^2}, -\frac{\bar x \alpha \lambda}{2}+\frac{\bar x\lambda}{2}+\frac{1}{2}\sqrt{{\bar x}^2\alpha^2\lambda^2-2{\bar x}^2\alpha\lambda^2-2\bar y\beta\lambda+{\bar x}^2\lambda^2}$. The first $p$ is $z$ (the fixed point), and the other two points are the period $2$ points. We name the last two points as $w_1$ and $w_2$ respectively ($w_1\leq w_2$). If we show that $s<w_2$, we are done. A direct calculation (with Python) shows that $s<w_2$ is equivalent to $\frac{9\bar y\beta}{4 {\bar x}^2(1-\alpha)^2}<\lambda$ (this is our assumption). 
\end{proof}

Now we assume $\frac{9\bar y \beta}{4{\bar x}^2(1-\alpha)^2}<\lambda$. (So $\Pi$ is a singleton.) Finally we show
\begin{lem}\label{ForthLem}
$f^3(s) > \max\{x\in \Pi\}$ if and only if $\frac{25\bar y\beta}{9{\bar x}^2(1-\alpha)}<\lambda$. 
\end{lem} 
\begin{proof}
This calculation is a bit involved, so we give some details. Let $z$ be the fixed point of $f$. Since we know that $\max\{x\in \Pi\}=\{z\}$, we have 
\begin{alignat*}{2}
f^3(s)-\max\{x\in \Pi\} &= f(f^2(s)) - z \\
&= f(f^2(s)) - f(z) \\
&= \left(f^2(s) + \lambda\left[\frac{\bar y\beta}{f^2(s)}-\bar x(1-\alpha)\right]\right) - \left(z + \lambda\left[\frac{\bar y\beta}{z}-\bar x(1-\alpha)\right]\right)\\
&=f^2(s)-z + \lambda\left[\frac{\bar y\beta}{f^2(s)}-\frac{\bar y\beta}{z}\right]\\
&=f^2(s)-z -\bar y\beta\lambda\left[\frac{f^2(s)-z}{zf^2(s)}\right]\\
&=(f^2(s)-z)\left(1-\frac{\bar y\beta\lambda}{f^2(s)z}\right)
\end{alignat*}
We consider the first term of the last expression. A direct calculation shows that $s>z$ if and only if $\lambda>\frac{\bar y\beta}{{\bar x}^2(1-\alpha)^2}$ (which we already assumed). So, we have $f^2(s)-z>f^2(s)-s>0$. (The last inequality followed from Lemma~\ref{ThirdLem}.) Next, a direct calculation shows that 
the second term of the last expression is positive if and only if $3{\bar x}^2(1-\alpha)^2\lambda-8\bar x(1-\alpha)\sqrt{\bar y\beta\lambda}+5\bar y\beta>0$. Now we see that under the condition $\frac{\bar y\beta}{{\bar x}^2(1-\alpha)^2}<\lambda$, this is equivalent to $\frac{25\bar y\beta}{9{\bar x}^2(1-\alpha)^2}<\lambda$.
\end{proof}
It is clear that $\min\{x\in \Pi\}=\{z\}$. So by the same argument, we obtain
\begin{lem}\label{FifthLem}
$f^3(s) \geq \min\{x\in \Pi\}$ if and only if $\frac{25\bar y\beta}{9{\bar x}^2(1-\alpha)}\leq\lambda$. 
\end{lem}
Note that $\frac{9\bar y\beta}{4{\bar x}^2(1-\alpha)^2} \approx \frac{2.25\bar y\beta}{{\bar x}^2(1-\alpha)^2}< \frac{2.78\bar y\beta}{{\bar x}^2(1-\alpha)^2}\approx \frac{25\bar y\beta}{9{\bar x}^2(1-\alpha)^2}$. By  Proposition~\ref{ChaosThm} and Lemmas~\ref{FirstLem}, \ref{ThirdLem}, \ref{SecondLem}, \ref{ForthLem}, \ref{FifthLem}, we have proved parts 2 and 3 of Theorem~\ref{main}. Finally, we are left to show (part 4 of Theorem~\ref{main}):
\begin{lem}\label{SixthLem}
If $\frac{25 \bar y\beta}{9 {\bar x}^2 (1-\alpha)^2} \leq \lambda < \frac{4\bar y\beta}{{\bar x}^2 (1-\alpha)^2}$, then the closed interval $E$ is attracting for $f$, that is, $f^n(p)\in E$ for some $n>0$.
\end{lem} 
\begin{proof}
Since $E=[a,b]$, we need to prove: (1)~if $0<p<a$, then $f^{n_1}(p)\in E$ for some $n_1\in \mathbb{N}$, (2)~if $b<p$, then $f^{n_2}(p)\in E$ for some $n_2\in \mathbb{N}$. Note that if $0<p<a$ (the first case), then $f(p)\geq f(s)=a$ since $f(s)$ is a global minimum for $f$, thus we just need to consider the second case. Let $p>b$. Then we have $f(p)-p=\lambda[\frac{\bar y\beta}{p}-\bar x(1-\alpha)]\leq \lambda[\frac{\bar y\beta}{b}-\bar x(1-\alpha)]=f(b)-b<0$. (The last strict inequality follows from $\frac{\bar y \beta}{{\bar x}^2(1-\alpha)^2}<\lambda$, see the proof of Lemma~\ref{FirstLem}) This shows that if $p>b$, in each iteration of $f$, the value of $p$ drops by $b-f(b)>0$ at least. So, to prove that $p$ is attracted to $E$, it is enough to show that the size of the drop is not too big (thus $p$ does not jump over $E$). Therefore, it is sufficient to have $p-f(p)\leq b-a$. Since this is equivalent to $\lambda[\bar x(1-\alpha)]\leq f^2(s)+s-f(s)$ and $s-f(s)>0$ under our assumption on $\lambda$, it suffices to show that $\lambda[\bar x(1-\alpha)]\leq f^2(s)$. After some computations, we obtain
\begin{equation*}\label{trapEqn}
f^2(s) - \lambda[\bar x(1-\alpha)]  = \frac{\lambda\left[3{\bar x}^2(1-\alpha)^2\lambda-8\bar x\sqrt{\bar y}(1-\alpha)\sqrt{\beta\lambda}+5\bar y \beta\right]}{\bar x(\alpha-1)\lambda + 2\sqrt{\bar y \beta \lambda}}. 
\end{equation*}
We check that in the last expression, the denominator is strictly positive if $0<\lambda<\frac{4\bar y\beta}{{\bar x}^2(1-\alpha)^2}$ (this follows from our assumption). Also, we see that the numerator is positive if $\frac{25\bar y \beta}{9 {\bar x}^2(1-\alpha)^2}\leq \lambda$. 

\end{proof}

\section{Proofs of Propositions~\ref{LiYorke} and~\ref{NoPeriodThree}} 
\begin{proof}[Proof of Proposition~\ref{LiYorke}]
The first part (this paragraph) of the following argument is a replicate of~\cite[Proof of Prop.~9.10]{Bhattacharya-book}. We include this to make the paper self-contained. Let $(\alpha,\beta)=(0.75,0.5)$, $(\bar x, \bar y)=(4,2)$, and $\lambda = 3.61$. Then we have $f(3.75)=1.9$, $f(1.9)=0.19$, $f(0.19)=15.58>3.75$. So, by the Li-Yorke theorem, there exists a period three cycle. Also, by~\cite[Prop.~9.10 and its proof]{Bhattacharya-book}, we know that the choices of $(\alpha, \beta)$ and $\lambda$ are robust (for this fixed $(\bar x, \bar y)$), so it is clear that there exist open sets $A$ and $C$ as in Proposition~\ref{LiYorke}. Thus, the only thing we need to show is that the choice of $(\bar x, \bar y)$ is also robust. This follows from the following numerical/graphical argument.  

In Figure~\ref{fig1}, we plot the graphs of $p_{t+1}=f(p_t)$ (dashed curve) and $p_{t+3}=f^3(p_t)$ (solid curve) together with the $45^\circ$ line. 
We see that the solid curve crosses the $45^\circ$ line at $p^*$ (the big dot) between $p=2$ and $p=3$. A numerical computation gives $p^{*}\cong 2.66$, and it is clear that $p^{*}$ is a point of period three. (From the picture we see that $p^{*}$ is not a fixed point.) Since the solid curve crosses (but not touching) the $45^\circ$ line at $p^{*}$ and $f^3$ is continuous in $\bar x$ and $\bar y$, a small perturbation of $\bar x$ and $\bar y$ does not affect the existence of a period three cycle. So, there exists an open set $B$ as required. (Actually, this graphical argument gives the existence of open sets $A$ and $C$ as well.) 
\end{proof}

\begin{proof}[Proof of Proposition~\ref{NoPeriodThree}]
Let $(\bar x, \bar y)=(4,2)$, $(\alpha,\beta)=(0.75,0.5)$, and $2.77< \lambda\leq 3.00$. Then by Theorem~\ref{main} it is clear that the map $f$ has an odd period cycle. Only thing we need to show now is that $f$ does not have a period three cycle. We use a graphical/numerical argument. In Figures~\ref{fig2},~\ref{fig3}, and~\ref{fig4} we plot the graphs of $f$ (dashed curve) and $f^3$ (solid curve) with the $45^{\circ}$ line. A general pattern is that if we increase $\lambda$, the graph of $f$ gets slightly deeper down, and $f^3$ becomes more "wavy". (Compare Figure~\ref{fig1} ($\lambda=3.61$) to Figure~\ref{fig2} ($\lambda=2.78$) for example.) For Figure~\ref{fig2} ($\lambda=2.78$), using a numerical computation, we have checked that the solid curve does not touch/cross the $45^{\circ}$ line except at $p=1$ (the fixed point of $f$). Thus, there is no period three cycle in this case. Likewise, for Figure~\ref{fig3} ($\lambda=3$), there is no period three cycle (although the solid curve seems touching the $45^\circ$ line but a numerical computation shows that it is not touching). However, if $\lambda=3.01$, a numerical computation shows that there exists a period three cycle. In Figure~\ref{fig4}, we plotted the $\lambda=3.05$ case since the $\lambda=3.01$ case is indistinguishable from the $\lambda=3$ case in picture. It is clear that the solid curve crosses the $45^\circ$ line near $p=4$ (that is not a fixed point of $f$).    



\begin{figure}[h!]
	\begin{center}
	\includegraphics[scale=0.45]{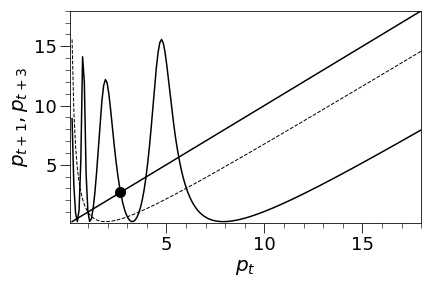} 
	\end{center}      
        \caption{$\lambda=3.61$}
        \label{fig1}
\end{figure}

\begin{figure}[h!]
	\begin{center}
	\includegraphics[scale=0.45]{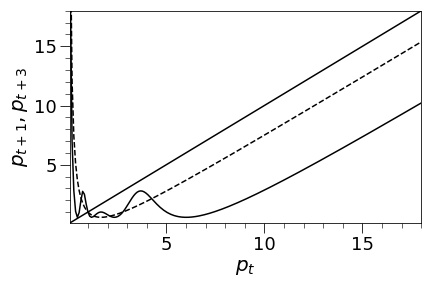}
	\end{center}
        \caption{$\lambda=2.78$}
        \label{fig2}
\end{figure}

\begin{figure}[h!]
	\begin{center}
	\includegraphics[scale=0.45]{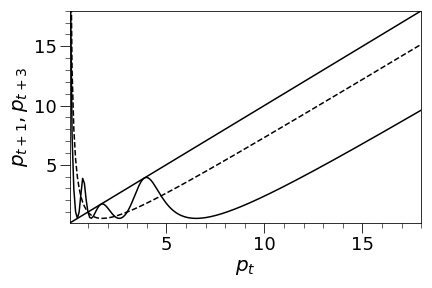}
	\end{center}
        \caption{$\lambda=3.0$}
        \label{fig3}
\end{figure}

\begin{figure}[h!]
	\begin{center}
	\includegraphics[scale=0.45]{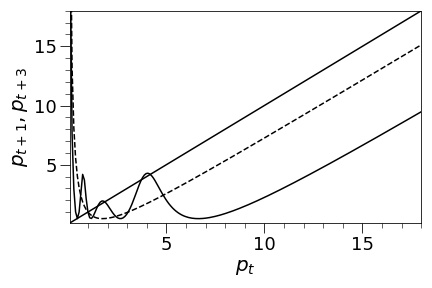}
	\end{center}
        \caption{$\lambda=3.05$}
        \label{fig4}
\end{figure}
\end{proof}

\section{Ergodic properties: an experimental approach}
Our (numerical/theoretical) argument in this and the next sections use ergodic theory. Here, we give a quick review of ergodic theory. If the reader is familiar with ergodic theory, skip Subsection~\ref{background}. Our basic references for ergodic theory are classical~\cite{ColletEckmann-dynamics-book},~\cite{Day-dynamics-book}, and~\cite{MeloStrien-dynamics-book}. Note that our strategy (philosophy) in this and the next sections stems from~\cite{Lyubich-forty-Modern} and~\cite{Shen:2014} (these are quite readable expository articles on recent developments of unimodal dynamics). We stress that a deep result by Avila et.al.~(Proposition~\ref{Avila}) theoretically supports our argument.  

\subsection{Background from ergodic theory}\label{background}
Let $I$ be a compact interval on $\mathbb R$.  Let $\mathcal{B}$ denote the Borel $\sigma$-algebra of $I$. Let $\zeta: \mathcal{B}\rightarrow [0,\infty]$ be a measure on $I$. A measurable map $g: I\rightarrow I$ is called \emph{ergodic with respect to $\zeta$} if whenever $g^{-1}(A)=A$ for $A\in \mathcal{B}$, then either $\zeta(A)=0$ or $\zeta(I\backslash A)=0$. We say that a measure $\zeta$ is \emph{g-invariant} if $\zeta(g^{-1}(A))=\zeta(A)$ for any $A\in \mathcal{B}$. We write $\mu$ for the Lebesgue measure on $I$. A measure $\zeta$ is called \emph{absolutely continuous with respect to $\mu$} if whenever $\mu(A)=0$ for $A\in \mathcal{B}$ then $\zeta(A)=0$. We write an "acim" for a measure that is absolutely continuous with respect to $\mu$ and $g$-invariant (if $g$ is clear from the context, we just say "invariant"). Remember that if $\zeta$ is an acim then there exists a ($\mu$-integrable) density function (Radon-Nikodym derivative) $\xi: I \rightarrow \mathbb{R}$ with $d\zeta = \xi d\mu$, see~\cite[\Rmnum{2}.8]{ColletEckmann-dynamics-book}. It is well-known that if $g$ is ergodic with respect to an acim $\zeta$, then we can give an estimate of $\xi$ (hence an estimate of $\zeta$) using some iterates of $g$, see~\cite[8.5.2 and 8.5.3]{Day-dynamics-book} (we use this argument below in Theorem~\ref{ergThm1}).  

Here, we recall (a special case of) the famous Birkhoff's ergodic theorem~\cite[Thm.~8.2]{Day-dynamics-book}, which states that under certain conditions the time average of $g$ is equal to its space average:
\begin{prop}\label{birkhoff}
If $g$ has an absolutely continuous invariant measure $\zeta$, $g$ is ergodic with respect to $\zeta$, and $g$ is $\zeta$-integrable, then
\begin{equation}\label{ergodicSumEqn}
\lim_{n\rightarrow \infty}\frac{1}{n}\sum_{k=0}^{n-1}g^{k}(x)=\int_I g\; d\zeta \text{ for $\zeta$-almost all $x$.}
\end{equation}
\end{prop} 
 
In particular, Proposition~\ref{birkhoff} says that (if the conditions are met) the time average of $g$ converges to a constant (for almost all $x$), in other words, we can "predict" the future on average. We call the integral on the right-hand side (or sometimes the sum inside the limit on the left-hand side) of Equation (\ref{ergodicSumEqn}), that is $\int_I g\; d\zeta$ (or sometimes $\frac{1}{n}\sum_{k=0}^{n-1}g^{k}(x)$), the \emph{ergodic sum} of $g$ (which we are referring to would be clear from the context). In the following, we try to apply Proposition~\ref{birkhoff} to our price dynamics defined by $\restr{f}{E}$ (with the condition $\frac{\bar y\beta}{{\bar x}^2(1-\alpha)^2} < \lambda < \frac{4\bar y\beta}{{\bar x}^2(1-\alpha)^2}$ to force that $p>0$ and $\restr{f}{E}\in\mathfrak{G}$). In the following, we write $f$ for $\restr{f}{E}$ to ease the notation. So we need to show that $f$ has an absolutely continuous invariant measure $\zeta$, $f$ is ergodic with respect to $\zeta$, and $f$ is $\zeta$-integrable. (The last condition is clear since $f$ is continuous and bounded on the compact interval $E$.)

\subsection{$S$-unimodal maps}
In general, it is pretty difficult to prove the existence of an acim for a measurable transformation $g$ except for some special cases such as when $g$ is "expansive" or "iteratively expansive" as studied in a classical paper of Lasota and Yorke~\cite{LasotaYorke-expansive-Trans}. Recall that $g$ is called \emph{expansive} if $g$ is piecewise $C^2$ and $|g'(x)|>1$ for $\mu$-almost all $x$. A typical example of an expansive map is a well-studied "tent map"~\cite[8.5.4 and 8.5.5]{Day-dynamics-book}. Further recall that a slightly more general "iteratively expansive" map, that is a piecewise $C^2$ map with $|g'(x)|^n>1$ for some positive integer $n>1$ for $\mu$-almost all $x$. From~\cite{LasotaYorke-expansive-Trans}, we know that a density function $\xi$ with $d\zeta=\xi d\mu$ is a fixed point of \emph{the Perron-Frobenius operator} $P$ from the set of measurable function on $I$ to itself, and that a uniform expansion of a map $g$ helps a lot to make $P$ a "nice" operator. As a result, it is not too hard to prove the existence of an acim in these (iteratively) expansive cases. See~\cite[Chap.~5]{MeloStrien-dynamics-book},~\cite[Sec.~4]{Shen:2014}, and~\cite{Lyubich-forty-Modern} for an overview of this problem and also see~\cite{SatoYano-ergo-AIP} and~\cite{SatoYano-chaos-IJET} for applications of an acim for an iteratively expansive map in economics. 

In this paper, our function $f$ is not (even iteratively) expansive since it has a critical point $s$. (So this is a hard case.) Thus, to establish the existence of an acim, we need some deep analytical results to investigate the counter play between the contraction of $f$ near the critical point $s$ and the expansion of $f$ at $f(s)$ (that is far from $s$). Now, we restrict the class of functions we consider to so-called $S$-unimodal maps (due to Singer~\cite{Singer-unimodal-SIAM}), whose ergodic properties are well studied, see~\cite[Part \Rmnum{2}]{ColletEckmann-dynamics-book},~\cite[Chap.~5]{MeloStrien-dynamics-book},~\cite{Avila-unimodal-Annals},~\cite{Avila-stochastic-Invent} for example. (Our function $f$ is actually $S$-unimodal as we will show below.) Let $g$ be a measurable transformation defined on a compact interval $I=[a,b]$ of $\mathbb{R}$.
\begin{defn}
A function $g$ is called \emph{$S$-unimodal} if the following conditions are satisfied:
\begin{enumerate}
\item{$g$ is $C^3$.}
\item{$g$ is unimodal with the unique critical point $x=c$ in $(a,b)$ and $g'(x)\neq 0$ except when $x=c$.}
\item{The Schwarzian derivative of $g$, that is, $Sg(x)=\frac{g'''(x)}{g'(x)}-\frac{3}{2}\left(\frac{g''(x)}{g'(x)}\right)^2$ is negative except at $x=c$.}
\end{enumerate}  
Moreover the critical point $x=c$ is called \emph{non-flat} and of order $l$ if there are positive constants $O_1$, $O_2$ with 
\begin{equation*}
O_1|x-c|^{l-1} \leq |g'(x)| \leq O_2|x-c|^{l-1}.
\end{equation*}
\end{defn}

Note that a well-known "logistic map" ($g(x)=r x(1-x)$ for $r\in (0,4]$) is $S$-unimodal. The first key result in this section is
\begin{prop}\label{ergodProp}
If $g$ is $S$-unimodal with a non-flat critical point and without an attracting periodic orbit, then $g$ is ergodic with respect to any absolutely continuous measure $\zeta$. 
\end{prop}
\begin{proof}
Let $g$ be $S$-unimodal with a non-flat critical point and without an attracting periodic orbit. Suppose that $g^{-1}(A)=A$ for some $A\in\mathcal{B}$. Then we have $\mu(g^{-1}(A))=\mu(A)$. From~\cite[Thm.~1.2]{MeloStrien-dynamics-book} we know that $g$ is ergodic with respect to $\mu$, so we obtain that $\mu(A)=0$ or $\mu(I\backslash A)=0$. This yields $\zeta(A)=0$ or $\zeta(I\backslash A)=0$ since $\zeta$ is absolutely continuous.  
\end{proof}

To end this subsection, we prove that
\begin{lem}\label{S-unimodalLem}
The function $f$ (restricted to $E=[f(\sqrt{\bar y\lambda\beta}),f^2(\sqrt{\bar y\lambda\beta})+\sqrt{\bar y\lambda\beta}]=[a,b]$) is $S$-unimodal and the unique critical point $c$ of $f$ is non-flat and of order $2$.
\end{lem}
\begin{proof}
First, we show that $f(p)=p+\lambda[\bar y \beta/p - \bar x(1-\alpha)]$ is $C^3$. We have $f'(p)=1-\lambda\bar y \beta/{p^2}$, $f''(p)=2\lambda\bar y \beta/{p^3}$, and $f'''(p)=-6\lambda\bar y \beta/{p^4}$. Since $p$ is positive (so it is not zero), it is clear that $f$ is $C^3$. Second, from Section~2, we know that $f$ is unimodal and has a unique critical point at $p=\sqrt{\lambda\bar y \beta}=s$ in $(a,b)$. It is easy to see that $f'(p)\neq 0$ except when $p=s$. Third, we have $Sf(p)=-6\lambda\bar y \beta/(p^2-\lambda\bar y\beta)^2<0$ except when $p=\sqrt{\lambda\bar y\beta}=s$. Finally, we show that the critical point $p=s$ is non-flat of order $l=2$. It is clear that $|f'(p)|$ is strictly concave and monotone increasing on the compact interval $[s,b]$. Also note that the righthand derivative of $|f'(p)|$ at $s$ is $2/\sqrt{\lambda\bar y\beta}$. So $\restr{f'(p)}{[s,b]}$ is bounded below by $f'(b)/(b-s)|p-s|$ and is bounded above by $2/\sqrt{\lambda\bar y\beta}|p-s|$. Likewise, we have that $|f'(p)|$ is strictly convex and monotone decreasing on $[a,s]$. Also, the lefthand derivative of $|f'(p)|$ is $-2/\sqrt{\lambda\bar y\beta}$. Therefore, $\restr{f'(p)}{[a,s]}$ is bounded below by $2/\sqrt{\lambda\bar y\beta}|p-s|$ and is bounded above by $f'(a)(s-a)|p-s|$. Thus we see that $f'(p)$ (on $E$) is bounded below by $\min\{f'(b)/(b-s), 2/\sqrt{\lambda\bar y\beta}\}|p-s|$ and is bounded above by $\max\{2/\sqrt{\lambda\bar y\beta}, f'(a)(s-a)\}|p-s|$.
\end{proof}

\subsection{Our strategy and the existence of an acim}
The second key result in this section is

\begin{prop}\cite[Thm.~\Rmnum{2}.4.1]{ColletEckmann-dynamics-book}\label{criticalOrbit}
If $g$ is $S$-unimodal, then every stable periodic orbit attracts at least one of $a$, $b$, or $c$ (i.e.~the endpoints of $I$ or the critical point of $g$). 
\end{prop}
Proposition~\ref{criticalOrbit} means that all "visible" orbits (in numerical experiments) are orbits containing $a$, $b$, or $c$ only (in the long run). We consider that only these visible orbits are meaningful in economics (or in real life) since it is widely believed that every economic modelling is some sort of an approximation of real economic activities and contains inevitable errors. We know that there are totally different point of view for economic modellings, but we do not argue here. Here is the third key result for this section: 
\begin{prop}\cite[Cor.~\Rmnum{2}.4.2]{ColletEckmann-dynamics-book}\label{noStableOrbit}
If $g$ is $S$-unimodal, then $g$ has at most one stable periodic orbit, plus possibly a stable fixed point. If the critical point $c$ is not attracted to a stable periodic orbit, then $g$ has no stable periodic orbit.   
\end{prop}

In this paper, we interpret our numerical calculations based on  Propositions~\ref{criticalOrbit} and~\ref{noStableOrbit}. In particular, we look at the orbit starting from the critical point $s$, that is $\{s, f(s), f^2(s),\cdots\}$ (we call this orbit the "critical orbit"). If the critical orbit seems to eventually converge to a periodic orbit, we conclude that we can see the future: the average price in the long run will be the average price in this attracting periodic orbit. Note that in this case, $f$ is neither ergodic nor has an acim since most $f^n(s)$ accumulate around this attracting periodic orbit, but we do not care (since we can still predict the future). Otherwise, we compute (or give an estimate for) the following \emph{Lyapunov exponent} at the critical point $p=s$ since the existence of a positive Lyapunov exponent at the critical point implies that the critical orbit is repelling and also is a strong indication for the existence of a chaos (hence the existence of an acim, see (CE1) in Proposition~\ref{sufficientACIM} below):

\begin{defn}
$\lim_{n\rightarrow \infty}\frac{1}{n}\sum_{i=1}^{n}\ln{|Dg^{n}(g(c))|}$ is called the \emph{Lyapunov exponent} of $g$ at $c$ (if the limit exists).
\end{defn}

If the Lyapunov exponent (at $c$) is positive, we test one of the following well-known sufficient conditions (within some numerical bound) to confirm the existence of an acim. 

\begin{prop}\label{sufficientACIM}
Suppose that $g$ is $S$-unimodal, $g$ has no attracting periodic orbit, and the critical point $c$ is non-flat. Then $g$ has a unique acim $\zeta$ and $g$ is ergodic with respect to $\zeta$ if one of the following conditions is satisfied:
\begin{enumerate}
\item{Collet-Eckmann conditions (together with some regularity conditions)~\cite[Sec.~1]{ColletEckmann-Liapunov-ETDS}:
\begin{alignat*}{2}
&\textup{(CE1)}\;\;\liminf_{n\rightarrow\infty}\frac{1}{n}\ln\left| Dg^n(g(c)) \right| > 0, \\
&\textup{(CE2)}\;\;\liminf_{n\rightarrow\infty}\frac{1}{n}\inf_{p\in g^{-n}(c)}\ln\left| Dg^n(p)) \right| > 0. 
\end{alignat*} 
} 
\item{Misiurewicz condition~\cite[Thms.~6.2 and~6.3]{Misiurewicz-acim-IHES}: The $\omega$-limit set of $c$ does not contain $c$, that is, 
\begin{equation*}
(MC)\;\; c\notin\bigcap_{n\geq 0}\overline{\{g^i(c): i\geq n\}}. 
\end{equation*}
}
\item{Nowicki-Van Strien summation condition~\cite[Main Thm.]{Strien-SC-Invent}: if $l$ is the order of $c$, then
\begin{alignat*}{2}
&(SC)\;\; \sum_{n=1}^{\infty}|Dg^n(g(c))|^{-1/l} < \infty.
\end{alignat*}
}
\end{enumerate}
\end{prop}
If one of the conditions in Proposition~\ref{sufficientACIM} is satisfied (within some numerical limitation), we conclude that we can predict the future by Proposition~\ref{birkhoff}. We must admit that our argument in this section is not rigorous (we hope to make it rigorous in the future), but we believe that we have provided enough (numerical/theoretical) evidence to support it. We stress that it is very hard to prove the existence of an acim for any non-expansive function $g$ (even for an $S$-unimodal $g$) by a rigorous analytic argument. There are only a few known examples of such, see a famous $g(x)=4x(1-x)$ example due to Ulam and Neumann~\cite{Ulam-ACIM-BullAMS}, also see~\cite[Sec.~7 Examples]{Misiurewicz-acim-IHES} for more examples. 

In this paper, we test Condition 3 (SC) in Proposition~\ref{sufficientACIM} since it is easy to compute (numerically) and covers the most general class of functions, see~\cite[4.2]{Shen:2014} for a comparison of these three sufficient conditions for the existence of an acim, also see~\cite[Chap.~\Rmnum{5}, Sec.~4]{MeloStrien-dynamics-book} for more on (SC). We found that (CE2) was hard to compute since the set $g^{-n}(c)$ can be very large for a large $n$. Also, the $\omega$-limit set of $c$ was difficult to compute for us although (MC) is theoretically beautiful. (For a numerical computation, it is not clear where to set the numerical bound to estimate the $\omega$-limit set.)

\begin{rem}
Roughly speaking, all three condtions (CE1), (MC), and (SC) are basically testing the same thing: they (more or less) guarantee that the critical orbit does not accumulate around the critical point $c$. (So, on the critical orbit, the expansion far from the critical point wins against the contraction near the critical point.)
\end{rem}

\subsection{$\lambda=3.61$ case}
For the rest of the paper, to simplify the exposition and to obtain sharp numerical results, we fix $(\alpha,\beta)=(0.75,0.5)$, $(\bar x, \bar y)=(4,2)$. Also, in this subsection, we fix $\lambda=3.61$ (in the next subsection, we let $\lambda$ vary). Then we have $E=[0.19, 17.48]$. In this case, the dynamics exhibit a Li-Yorke chaos as shown in the previous sections. It is easy to see that a similar analysis (as follows) can be done for any $\alpha, \beta, \bar x, \bar y, \lambda$. 

Here is our main result in this section. (We can predict the future even if $f$ is chaotic.)

\begin{thm}\label{ergThm1-1}
There exists an acim $\zeta$ (whose estimate is as in Figure~\ref{fig8}) for $f$ on $E$. Moreover, we have $\lim_{n\rightarrow \infty}\frac{1}{n}\sum_{k=0}^{n-1}f^{k}(p)\approx 4.5$ for $\zeta$-almost all $p\in E$.
\end{thm} 
A few comments are in order. Although we give an estimate of $\zeta$ in Figure~\ref{fig8}, it is hard to give an explicit formula for $\zeta$ (thus it is hard to express $\zeta$ in a concrete way). Also, looking at Figure~\ref{fig8} (and its numerical data), we see that the density of $p$ is positive in $[0.19, 15.58]=[f(s), f^2(s)]$ and is zero in $[15.58, 17.48]=[f^2(s),f^2(s)+s]$. Thus, without a loss, if one wishes, Theorem~\ref{ergThm1-1} can be stated (in a more readable form) as follows: 

\begin{thm}\label{ergThm1}
$\lim_{n\rightarrow \infty}\frac{1}{n}\sum_{k=0}^{n-1}f^{k}(p)\approx 4.5$ for $\mu$-almost all $p\in \tilde{E}:=[f(s),f^2(s)]=[0.19, 15.58]$.  
\end{thm}

Now we start looking at the model closely. First, following our strategy as in the last subsection, we consider the critical orbit of $f$. Note that the critical point of $f$ is $s=\sqrt{\bar y \lambda \beta}=1.9$. Using the first $100$ iterates of $f^n(s)$, we obtain Figure~\ref{fig5} that shows a chaotic behaviour of the iterates of $f$. 

\begin{figure}[h!]
	\begin{center}
    	\includegraphics[scale=0.6]{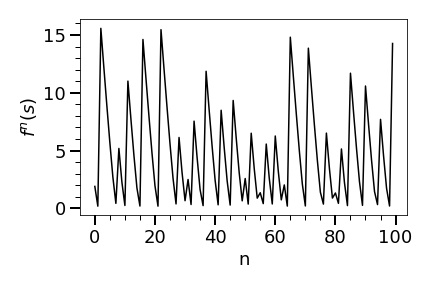}
	\end{center}
    \caption{The first $100$ iterates of $f^n(s)$ look chaotic}\label{fig5}
\end{figure}

In particular, it seems that $f$ has no attracting periodic orbit. To convince the reader that this really the case, we give an estimate for the Lyapunov exponent (using the first 10000 terms of $f^n(s)$). Figure~\ref{fig6} shows that the first 1000 terms are enough to estimate the Lyapunov exponent (but we used the first 10000 terms to be safe). We obtain
\begin{lem}
$\lim_{n\rightarrow \infty}\frac{1}{n}\ln{|Df^{n}(f(s))|}\approx \frac{1}{10000}\ln{|Df^{10000}(f(s))|}\approx 0.5882>0$.
\end{lem}  
\begin{rem}
Let $c_0=s, c_1=f(c_0), c_2=f(c_1), \cdots, c_{n}=f(c_{n-1})$. Note that to compute the Lyapunov exponent we have used the chain rule, that is, $Df^n(f(s))=f'(c_1)\times f'(c_2) \times f'(c_3) \times \cdots \times f'(c_{n})$ (so it is not hard to compute). 
\end{rem}

\begin{figure}[h!]
	\begin{center}
    	\includegraphics[scale=0.6]{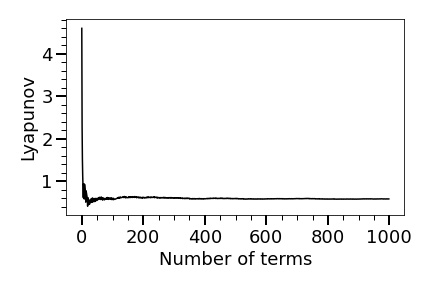}
	\end{center}
    \caption{Convergence of the Lyapunov exponent}\label{fig6}
\end{figure}

Now we check the Nowicki-Van Strien summation condition (SC). Figure~\ref{fig7} shows that the sum in (SC) stabilises if we use the first 100 terms. To be safe, we use 1000 terms to estimate the infinite sum in (SC). We obtain 

\begin{figure}[h!]
	\begin{center}
    	\includegraphics[scale=0.6]{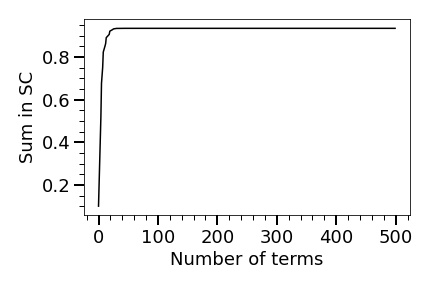}
	\end{center}
    \caption{Convergence of the sum in SC}\label{fig7}
\end{figure}

\begin{lem}\label{SCLem}
$\sum_{n=1}^{\infty}|Df^n(f(s))|^{-1/l}\approx \sum_{n=1}^{1000}|Df^n(f(s))|^{-1/l}\approx 0.935065399839560$.
\end{lem}
\begin{rem}
Since Lemma~\ref{SCLem} is crucial for our argument, we have double-checked (SC) with 100000 terms obtaining  $\sum_{n=1}^{100000}|Df^n(f(s))|^{-1/l}\approx 0.935065399839560$. (This is the same number as in Lemma~\ref{SCLem}!.)
\end{rem}
Now we conclude that $f$ has a unique acim $\zeta$ and $f$ is ergodic with respect to $\zeta$ by Proposition~\ref{sufficientACIM}. Next, in Figure~\ref{fig8} using $f^n(c)$ with $n$ from $1000$ to $10000$ (after removing the effect of a transient period) we obatin an estimate of the density function (Radon-Nikodym derivative) $\xi: E \rightarrow \mathbb{R}$ with $d\zeta = \xi d\mu$. 

\begin{figure}[h!]
	\begin{center}
    	\includegraphics[scale=0.6]{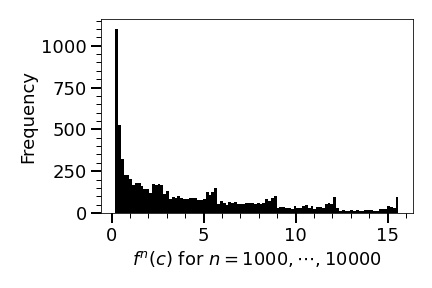}
	\end{center}
    \caption{Estimate of the density function}\label{fig8}
\end{figure}

Next, we directly compute the ergodic sum of $f^n(p)$ using the initial $p=s$ and the first $100000$ terms of $f^n(p)$. We get
\begin{lem}
 $\lim_{n\rightarrow \infty}\frac{1}{n}\sum_{k=0}^{n-1}f^{k}(s)\approx \frac{1}{100000}\sum_{k=0}^{99999}f^{k}(s)\approx 4.483627795147089$. 
 \end{lem}
Note that Figure~\ref{fig9} shows that the ergodic sum converges if we use more than $2000$ terms to estimate it (we use $100000$ terms to be safe).  

\begin{figure}[h!]
	\begin{center}
    	\includegraphics[scale=0.6]{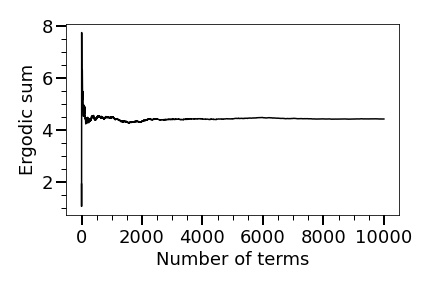}
	\end{center}
    \caption{Convergence of the ergodic sum of $f^n(c)$}\label{fig9}
\end{figure}

To end this section, we compute the ergodic sums using various initial values. Figure~\ref{fig10} shows the result where the initial $p$ is taken from $0.19, 0.20, 0.21, \cdots, 17.48$ (excluding $p=1$ (the unique fixed point of $f$). By Figure~\ref{fig10} (where each ergodic sum is estimated using the first $10000$ terms of $f^n(p)$), we conclude that Theorems~\ref{ergThm1-1} and~\ref{ergThm1} hold and the ergodic sums of $f^n(p)$ converge to somewhere around $4.5$ for $\zeta$ (or $\mu$) almost all $p\in E$ (or $p\in \tilde{E}$). 

\begin{figure}[h!]
	\begin{center}
    	\includegraphics[scale=0.6]{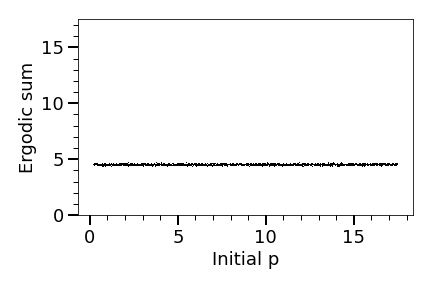}
	\end{center}
    \caption{Ergodic sums of $f^n(p)$ using various initial $p$}\label{fig10}
\end{figure}





\section{Ergodic properties: a sensitivity analysis}
\subsection{Chaos is not too bad}
In this final section, we still keep $(\alpha,\beta)=(0.75,0.5)$, $(\bar x, \bar y)=(4,2)$, but we let $\lambda$ vary. Recall that, from Subsection~2.3, we know that there exists an odd period (hence a Li-Yorke chaos) for $2.77<\lambda<4$. We conduct a sensitivity analysis, that is, how the ergodic sums vary when $\lambda$ changes with $1<\lambda<4$. We know that the unique fixed point of $f$ is $p=\frac{\bar y \beta}{\bar x (1-\alpha)}=1$ (independent of $\lambda$), and the (unique) critical point of $f$ is $s=\sqrt{\bar y \lambda \beta}=\sqrt{\lambda}$ (dependent of $\lambda$). 

Now, using the same strategy as in the last section, we investigate how the critical orbit $\{s, f(s), f^2(s),\cdots\}$ behaves for each $\lambda$. Figure~\ref{fig11} is the bifurcation diagram for $f$. We (roughly) see that: (1)~for $\lambda<2$, $f^n(s)$ converges to the unique attracting fixed point (that is $p=1$), (2)~for $2<\lambda<2.5$, $f^n(s)$ converges to a period-$2$ orbit, (3)~for $2.5<\lambda<2.7$ (roughly), period doubling bifurcations occur and $f^n(s)$ converges to a period-$4$ (8, 16, and so on) orbit, (4)~for $2.7<\lambda<3$  (possibly starting at around $2.77$ based on results in previous sections), we see a chaos except a few "windows", (5)~at $\lambda=3$ (or $3.01$), we see a period three orbit, (6)~for $3<\lambda<3.1$ (roughly), again, we see period doubling bifurcations and $f^n(s)$ converges to an orbit of period $3\cdot 2, 3\cdot 2^2, 3\cdot 2^3$ and so on, (7)~for $3.1<\lambda$, we see a chaos again except a few windows. 

\begin{rem}
Actually, we have obtained the bifurcation diagram of $f$ for $1<\lambda<4$, but In Figure~\ref{fig11}, we have chopped the diagram at $\lambda=3.5$. The reason is that for $3.5<\lambda<4$, the maximum value of $f^n(p)$ grows exponentially, so the diagram gets very skewed (and the part of the diagram shown in Figure~\ref{fig11} becomes almost invisible). For $3.5<\lambda<4$ we have obtained a chaotic region (with some windows). 
\end{rem}

\begin{figure}[h!]
	\begin{center}
    	\includegraphics[scale=0.7]{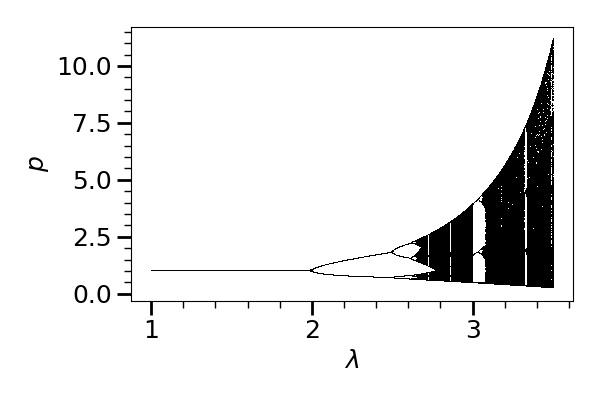}
	\end{center}
    \caption{The bifurcation diagram of $f$}\label{fig11}
\end{figure}

Next, we compute Lyapunov exponents of $f$ at $s$ for various $\lambda$ and obtain Figure~\ref{fig12}. We see that roughly speaking the Lyapunov exponent is negative for $1<\lambda<2.7$ (thick dots), and positive for $2.7<\lambda<4$ (thin dots) except a few "windows" corresponding to the windows in the bifurcation diagram (Figure~\ref{fig11}).

\begin{figure}[h!]
	\begin{center}
    	\includegraphics[scale=0.6]{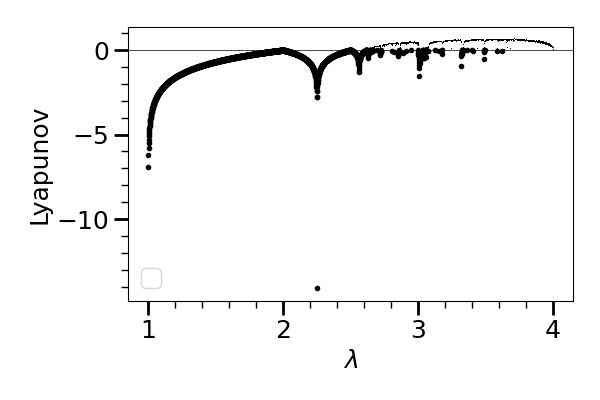}
	\end{center}
    \caption{Lyapunov exponents of $f$}\label{fig12}
\end{figure}

Now, we test (SC) for $2.7<\lambda<4$ (we do not care for $\lambda<2.7$ since it is a non-chaotic region and easy to predict the future). We obtain Figure~\ref{fig13}.

\begin{figure}[h!]
	\begin{center}
    	\includegraphics[scale=0.6]{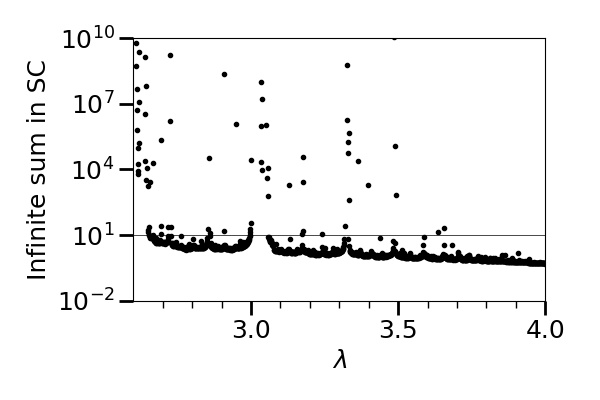}
	\end{center}
    \caption{Infinite sums in SC}\label{fig13}
\end{figure}

Our argument is based on a numerical computation using $1000$ terms to estimate the infinite sum in (SC), so we need to decide when we conclude that the infinite sum is finite. We draw the (ad hoc) line at $10$ (that is the horizontal line in Figure~\ref{fig13}). Basically, we only avoid $\lambda$ where (the estimate of) the infinite sum grows exponentially. We hope to rigorously prove that the infinite sums here are finite in the future work. 

Here are main results in this section.
\begin{thm}\label{goodLThm}
For $2.75<\lambda<4$ except $\lambda$ values corresponding to the few windows in Figure~\ref{fig11} (and possibly except some $\lambda$ values whose total Lebesgue measure is $0$, see Proposition~\ref{Avila} below), there exists a unique acim for $f$. Moreover for these $\lambda$ values, the ergodic sums of $f$ are as in Figure~\ref{fig14}.  
\end{thm}

For $\lambda$ values as in Theorem~\ref{goodLThm} (satisfying the SC), we obtain a pretty smooth relation between $\lambda$ and the ergodic sums of $f$ as in Figure~\ref{fig14} (using $5000$ terms to estimate the ergodic sums). Extending Theorem~\ref{goodLThm} (and Figure~\ref{fig14}) using the naive estimates of the ergodic sum (that is $\sum_{k=0}^{9999}f^k(s)$) for $\lambda$ that does not satisfy (SC), we obtain Theorem~\ref{finalThm} (and Figure~\ref{fig16}).

\begin{figure}[h!]
	\begin{center}
    	\includegraphics[scale=0.6]{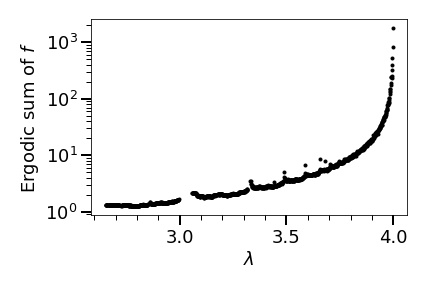}
	\end{center}
    \caption{Ergodic sums of $f$}\label{fig14}
\end{figure}

\subsection{Comments and interpretations on Theorem~\ref{finalThm}}

Here, we add a few comments and interpretations on Theorem~\ref{finalThm}: (1)~We found it surprising that the overall behaviour of the ergodic sums of $f$ is quite smooth and stable considering the fact that as $\lambda$ increases $f^n(x)$ go through a quite a bit of qualitative changes from a stable fixed point, attracting periodic orbit of different periods, and finally chaotic behaviours. (2)~The gradual increase of the ergodic sum of $f$ as $\lambda$ increases was unexpected (for us). We knew that as $\lambda$ increases, $f$ takes more extreme values (very high/very low), but knew nothing about their distributions. To give some (mathematical) explanation for the behaviour (2), we obtained the estimates of distributions of $f^n(p)$ for various $\lambda$ using $10000$ terms in Figures~\ref{fig16},~\ref{fig17}, and~\ref{fig18} (see Figure~\ref{fig8} also).

Our computation shows that: (1)~The "shapes" of the distributions of $f^n(s)$ for various $\lambda$ look similar: the density is high for a low range and low for a high range. (2)~$f$ takes more extreme values as $\lambda$ becomes large, however, the extension of the upper bound is much greater than that of the lower bound (meaning that for $E=[f(s), f^2(s)+s]$, as $\lambda$ becomes large, $f(s)$ gets smaller a little bit, but $f^2(s)+s$ gets larger quite a bit), as shown in Figures~\ref{fig16},\ref{fig17},\ref{fig18}, and~\ref{fig8}, (3)~The distribution gets smoother as $\lambda$ becomes large. 

We conclude that the gradual increase of the ergodic sum of $f$ happens as $\lambda$ increases because: (1)~$f$ takes more extreme values with a great extension of the upper bound (and with a little extension of the lower bound) as $\lambda$ increases, (2)~The shapes of the distributions do not change much as $\lambda$ increases (although it gets smoother), (3)~Thus, the average of $f^n(s)$ goes up as $\lambda$ increases. We do not know why the overall behaviour of ergodic sums of $f$ is so smooth and stable. Also, we do not know why the density curve gets smoother as $\lambda$ increases. We leave these issues for a future work.    

\begin{figure}[h!]
	\begin{center}
    	\includegraphics[scale=0.6]{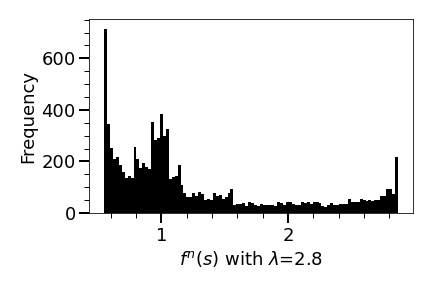}
	\end{center}
    \caption{Density of $f^n(s)$ with $\lambda=2.8$}\label{fig16}
\end{figure}

\begin{figure}[h!]
	\begin{center}
    	\includegraphics[scale=0.6]{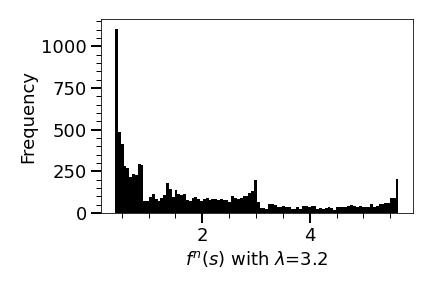}
	\end{center}
    \caption{Density of $f^n(s)$ with $\lambda=3.2$}\label{fig17}
\end{figure}

\begin{figure}[h!]
	\begin{center}
    	\includegraphics[scale=0.6]{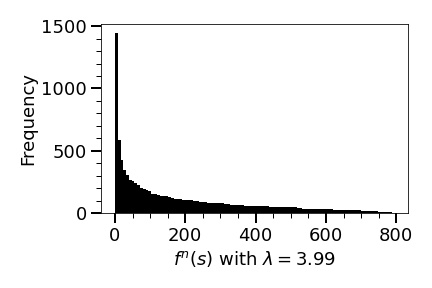}
	\end{center}
    \caption{Density of $f^n(s)$ with $\lambda=3.99$}\label{fig18}
\end{figure}

Theorem~\ref{finalThm} (and the whole results in this paper) says that a naive estimate of the ergodic sums of $f$ (estimate of the future) using a reasonably large number of terms (say $5000~10000$ terms) is not too bad. To end the paper, we quote a deep result of Avila (2014 fields medalist) and others~\cite[Sec.~3.1, Theorem B]{Avila-stochastic-Invent} that supports Theorem~\ref{finalThm}. 
\begin{prop}\label{Avila}
In any non-trivial real analytic family of quasiquadratic maps (that contains $S$-unimodal maps), (Lebesgue) almost any map is either regular (i.e., it has an attracting cycle) or stochastic (i.e., it has an acim). 
\end{prop}
\begin{rem} A technical note: "non-triviality" is guaranteed for our set of maps $f$ (parametrised by $\lambda$) since there exist two maps in this set that are not topologically conjugate. For example, take $f$ with $\lambda=1.5$ (non-chaotic) and $f$ with $\lambda=3.61$ (chaotic). See~\cite[Sec.~2.8, Sec.~2.13, Sec.~3.1]{Avila-stochastic-Invent} for the precise definitions of "non-trivial" and "quasiquadratic" (those are a bit too technical to state here). Also, see~\cite{Avila-unimodal-Annals}, and~\cite{Lyubich-forty-Modern} for more on this.
\end{rem}
\begin{rem}
We need "Lebesgue almost" (or "except a set of measure zero") in Theorems~\ref{finalThm},~\ref{goodLThm}, and Proposition~\ref{Avila} since the following (anomalous) examples are known, see~\cite{Hofbauer-asymptotic-MathPhys} and~\cite{Johnson-singular-MathPhys}: for a quadratic map $T_\lambda(x)=\lambda x(1-x)$ (parametrised by $\lambda$), there exists $\lambda$ such that $T_\lambda$ does not have an attracting periodic orbit and shows a chaotic behaviour, but does not have an acim. We expect that for our $f$, we obtain examples of the same properties (although we have not checked yet). The point is that we do not care such anomalous cases since the $\lambda$ values corresponding to such examples are of Lebesgue measure zero and our approach in this (and the last) section is probabilistic.  
\end{rem}

\section*{Acknowledgements}
This research was supported by a JSPS grant-in-aid for early-career scientists (22K13904) and an Alexander von Humboldt Japan-Germany joint research fellowship. 
\bibliography{econbib}

\end{document}